\documentclass[11pt]{article}  
\usepackage{amsfonts}
\usepackage{fullpage}
\usepackage{latexsym}
\usepackage{amssymb,amsmath}
\usepackage{microtype}
\usepackage{hyperref}
\usepackage{color}

\def\geq{\geqslant}
\def\leq{\leqslant}

\def\IR{\mathbb{R}}

\newtheorem{theorem}{Theorem}
\newtheorem{lemma}[theorem]{Lemma}
\newtheorem{problem}{Problem}

\def\boxit#1{\vbox{\hrule\hbox{\vrule\kern3pt
  \vbox{\kern3pt#1\kern3pt}\kern3pt\vrule}\hrule}}
\def\Box{\boxit{\null}}
\newcommand{\qed}{~$\Box$\medbreak}
\newenvironment{proof}{\noindent{\bf Proof: }}{\qed}

\begin{document}

\bibliographystyle{plain}

\title{Computing the Gromov hyperbolicity of a discrete metric space}
\author{Herv\'e Fournier
\thanks{Univ Paris Diderot, Sorbonne Paris Cit\'e,
Institut de Math\'ematiques de Jussieu, UMR 7586 CNRS, F-75205 Paris, France.
Email: \texttt{fournier@math.univ-paris-diderot.fr}.}
\and
Anas Ismail
\thanks{King Abdullah University of Science and Technology (KAUST),
Geometric Modeling and Scientific Visualization Center,
Thuwal 23955-6900. Saudi Arabia.
Email: \texttt{anas.ismail@kaust.edu.sa}.}
\and
Antoine Vigneron
\thanks{King Abdullah University of Science and Technology (KAUST),
Geometric Modeling and Scientific Visualization Center,
Thuwal 23955-6900. Saudi Arabia.
Email: \texttt{antoine.vigneron@kaust.edu.sa}.}
}

%\date{}
\maketitle

\begin{abstract}
We give exact and approximation algorithms for computing the 
Gromov hyperbolicity of an $n$-point discrete metric space.
We observe that computing the Gromov hyperbolicity from a 
fixed base-point reduces to a (max,min) matrix product.
Hence, using the (max,min) matrix product algorithm
by Duan and Pettie, the fixed base-point hyperbolicity
 can be determined in $O(n^{2.69})$ time.
It follows that the Gromov hyperbolicity can be computed in
$O(n^{3.69})$ time, and a 2-approximation can be found
in $O(n^{2.69})$ time. We also give a $(2\log_2 n)$-approximation
algorithm that runs in $O(n^2)$ time, based on a tree-metric
embedding by Gromov. We also show that hyperbolicity at a fixed base-point
cannot be computed in $O(n^{2.05})$ time, unless there exists a faster 
algorithm for (max,min) matrix multiplication than 
currently known.
\end{abstract}

\smallskip
\noindent \textbf{Keywords.} Algorithms design and analysis; Approximation algorithms; 
Discrete metric space; Hyperbolic space; (max,min) matrix product. 

%\newpage

\section{Introduction}

Gromov introduced a notion of metric-space hyperbolicity~\cite{BonkS00,Gromov87} 
using a simple four point condition. (See Section~\ref{sec:Gromov}.)
This definition is very attractive from a computer scientist point of
view as the hyperbolicity of a finite metric space
can be easily computed by brute force, by simply
checking the four point condition at each quadruple of points. However,
this approach takes $\Theta(n^4)$ time for an $n$-point metric space,
which makes it impractical for some applications to 
networking~\cite{cohen:hal-00735481}.
Knowing the hyperbolicity is important, as the running time and space requirements of 
previous algorithms designed for Gromov hyperbolic spaces are often 
analyzed in terms of their Gromov 
hyperbolicity~\cite{Chepoi12,ChepoiE07,KrauthgamerL06}.
So far, it seems that no better algorithm than brute force
was known for computing the Gromov hyperbolicity~\cite{Chen12}.
In this note, we give faster exact and approximation algorithms based 
on previous work on (max-min) matrix products by Duan and 
Pettie~\cite{DuanP09}, and the tree-metric embedding by Gromov~\cite{Gromov87}.

The exponent of matrix multiplication $\mu$ is the infimum
of the real numbers $\omega > 0$ such
that two  $n\times n$ real matrices can be multiplied in
$O(n^\omega)$ time,  exact arithmetic operations being performed
in one step~\cite{vzGG03}. Currently,
$\mu$ is known to be less than 2.373~\cite{Williams12}.
In the following,
$\omega$ is a real number such that we can multiply
two $n\times n$ real matrices in $O(n^\omega)$ time.

Our algorithm for
computing the Gromov hyperbolicity runs in $O(n^{(5+\omega)/2})$ time, which
is $O(n^{3.69})$. (See Section~\ref{sec:exact}.) For a fixed base-point,
this improves to $O(n^{(3+\omega)/2})$, which also yields a 2-factor
approximation for the general case within the same time bound.
(See Section~\ref{sec:factor2}.)
We also give a quadratic-time $(2\log_2 n)$-approximation algorithm.
(See Section~\ref{sec:factorlog}.)
Finally, we show that hyperbolicity at a fixed base-point cannot be computed
in time $O(n^{3(\omega-1)/2})=O(n^{2.05})$, unless (max,min) matrix
product can be computed in time $O(n^\tau)$ for $\tau< (3+\omega)/2$.
(See Section~\ref{sec:LowerBound}.)
The currently best known algorithm runs in $O(n^{(3+\omega)/2})$ 
time~\cite{DuanP09}.

\subsection{Gromov hyperbolic spaces}\label{sec:Gromov}

An introduction to Gromov hyperbolic spaces can be found in the article by
Bonk and Schramm~\cite{BonkS00}, and in the book by Ghys and de la Harpe~\cite{Ghys90}. 
Here we briefly present some definitions and facts that will be needed in this note.

A metric space $(M,d)$ is said to be $\delta$-hyperbolic for some $\delta \geq 0$
if it obeys the so-called {\em four point condition}: 
For any $x,y,z,t \in M$, the largest two distance sums among $d(x,y)+d(z,t)$,
$d(x,z)+d(y,t)$, and $d(x,t)+d(y,z)$, differ by at most $2\delta$. 
The {\em Gromov hyperbolicity} $\delta^*$ of $(M,d)$ is the smallest 
$\delta^* \geq 0$ such that $(M,d)$ is $\delta^*$-hyperbolic.

For any $x,y,r \in M$, the {\em Gromov product}
of $x,y$ at $r$ is defined as
\[(x \vert y)_r=\frac{1}{2} \left(d(x,r)+d(r,y)-d(x,y)\right).
\]
The point $r$ is called the {\em base point}.
Gromov hyperbolicity can also be defined in terms of the Gromov product,
instead of the four point condition above. The two definitions are equivalent,
with the same values of $\delta$ and $\delta^*$.
So a metric space $(M,d)$ is $\delta$-hyperbolic if and only if, for any $x,y,z,r \in M$
\[(x \vert z)_r \geq \min\{(x \vert y)_r,(y \vert z)_r\}-\delta.
\]
The Gromov hyperbolicity $\delta^*$ is the smallest value of $\delta$ 
that satisfies the above property. In other words,
\begin{equation*}
\delta^*=\max_{x,y,z,r}\left\{ \min\{(x \vert y)_r,(y \vert z)_r\}-(x \vert z)_r \right\}.
\end{equation*}
The hyperbolicity $\delta_r$ at base point $r$ is defined as
\begin{equation}\label{eq:fixedbase}
\delta_r=\max_{x,y,z}\left\{ \min\{(x \vert y)_r,(y \vert z)_r\}-(x \vert z)_r \right\}.
\end{equation}
Hence, we have
\begin{equation}\label{eq:deltadelta}
\delta^*=\max_r \delta_r.
\end{equation}

\section{Algorithms}
In this section, we consider a discrete metric space $(M,d)$ with
$n$ elements, that we denote $x_1,\ldots,x_n$. Our goal is to compute
exactly, or approximately, its hyperbolicity $\delta^*$, or 
its hyperbolicity $\delta_r$ at a base point $r$.

\subsection{Exact algorithms}\label{sec:exact}
The (max,min)-product $A \otimes B$ of two real matrices $A,B$ is 
defined as follows:
\[(A \otimes B)_{ij}=\max_k \min \{A_{ik},B_{kj}\}.
\]

Duan and Pettie~\cite{DuanP09} gave an $O(n^{(3+\omega)/2})$-time
algorithm for 
computing the (max,min)-product of two $n\times n$ matrices.

Let $r$ be a fixed base-point. 
By Equation~\eqref{eq:fixedbase},
if $A$ is the matrix defined by $A_{ij}=(x_i \vert x_j)_r$ for any $i,j$,
then $\delta_r$ is simply the largest coefficient in
$(A \otimes A) - A$. So we can compute $\delta_r$ in 
$O(n^{(3+\omega)/2})$ time. Maximizing over all values of $r$,
we can compute the hyperbolicity $\delta^*$ in
$O(n^{(5+\omega)/2})$ time, by Equation~\eqref{eq:deltadelta}.

\subsection{Factor-2 approximation}\label{sec:factor2}

The hyperbolicity $\delta_r$ with respect to any base-point is
known to be a 2-approximation of the hyperbolicity $\delta^*$~\cite{BonkS00}.
More precisely, we have $\delta_r \leq \delta^* \leq 2 \delta_r$. So, using the 
algorithm of Section~\ref{sec:exact}, we can pick an arbitrary base-point $r$ and
compute $\delta_r$ in $O(n^{(3+\omega)/2})$ time,
which gives us a 2-approximation of $\delta^*$.

\subsection{Logarithmic factor  approximation}\label{sec:factorlog}

Gromov~\cite{Gromov87} 
(see also the article by Chepoi et al.~\cite[Theorem 1]{Chepoi12}
and the book by Ghys and de la Harpe~\cite[Chapter 2]{Ghys90}) showed 
that any $\delta$-hyperbolic metric space $(M,d)$ can be embedded into 
a weighted tree $T$ with an additive error $2\delta \log_2 n$, and
this tree  can be constructed in time $O(n^2)$.
In particular, if we denote by
$d_T$ the metric corresponding to such a tree $T$, then 
\begin{equation}\label{eq:approxtree}
d(a,b)-2 \delta^* \log_2 n \leq d_T(a,b) \leq d(a,b) \mbox{ for any } a,b \in M. 
\end{equation}
This construction can be performed without prior knowledge of $\delta^*$.

We compute $D=\max_{a,b \in M} d(a,b)-d_T(a,b)$ in time $O(n^2)$. We claim
that:
\begin{equation}\label{eq:lognapx}
\delta^* \leq D \leq 2\delta^* \log_2 n.
\end{equation}
So we obtain a $(2\log_2 n)$-approximation $D$ of $\delta^*$ in time $O(n^2)$.

We still need to prove the double inequality~\eqref{eq:lognapx}. It follows 
from Equation~\eqref{eq:approxtree} that 
$d(a,b)-d_T(a,b) \leq 2 \delta^* \log_2 n$ for any $a,b$, and thus 
$D \leq 2 \delta^* \log_2 n$. In the following, we prove the other inequality.

For any $x,y,z,t$, we
denote by $\delta(x,y,z,t)$  the difference between the two largest 
distance sums among $d(x,y)+d(z,t)$, $d(x,z)+d(y,t)$, and $d(x,t)+d(y,z)$.
Thus, if for instance $d(x,y)+d(z,t) \geq d(x,z)+d(y,t) \geq d(x,t)+d(y,z)$,
we have $\delta(x,y,z,t)=d(x,y)+d(z,t)-d(x,z)-d(y,t)$. 
We also need to introduce the difference $\delta_T(x,y,z,t)$  between the two largest  
sums among $d_T(x,y)+d_T(z,t)$, $d_T(x,z)+d_T(y,t)$, and $d_T(x,t)+d_T(y,z)$.

For any $a,b \in M$, we have $d(a,b)-D \leq d_T(a,b) \leq d(a,b)$, so 
$\delta(x,y,z,t)-\delta_T(x,y,z,t) \leq 2D$, because in the worst case, the
largest sum with respect to $d$ is the same as the largest sum with respect
to $d_T$, and the second largest sum with respect to $d_T$ is equal to
the second largest sum with respect to $d$ minus $2D$.
But by construction, $d_T$ is a tree metric~\cite{Chepoi12}, so 
$\delta_T(x,y,z,t)=0$ for any $x,y,z,t$. Therefore $\delta(x,y,z,t)\leq 2D$
for any $x,y,z,t$, which means that $\delta^* \leq D$. 

\section{Conditional lower bounds}\label{sec:LowerBound}

We show that computing hyperbolicity at a fixed base-point is intimately connected with (max,min)-product. From the previous section, any improvement on the complexity of (max,min)-product yields an improvement on our algorithm to compute hyperbolicity. We show that a partial converse holds: Any improvement on the complexity of computing hyperbolicity at a fixed base-point
below $n^{3(\omega-1)/2}$ would give an improved algorithm for computing the (max,min)-product.

We consider the following decision problem.

\begin{problem}[{\sc hyp}: Fixed-base hyperbolicity]\label{pb:hyp}
Given a metric on $n$ points, a point $r$ and $\alpha \geq 0$, decide if the hyperbolicity $\delta_r$ at base point $r$ is larger than $\alpha$.
\end{problem}

Note that we do not ask to check whether the input is indeed a metric (no subcubic
algorithm is known for this problem~\cite{WilliamsW10}).
A tree metric is a metric such that $\delta_r=0$ for some base point $r$ or,
equivalently, such that $\delta_r=0$ for any base point $r$, so the special 
case $\alpha=0$ can be solved in $O(n^2)$ time as a tree metric 
can be recognized in $O(n^2)$ time~\cite{Bandelt90}.
In this section, we show that we cannot get a quadratic time algorithm
(or even an $O(n^{2.05})$ time algorithm)
unless some progress
is made on the complexity of (max,min) matrix product.
Our main tool is a result from Vassilevska and Williams~\cite{WilliamsW10} 
stated below for the special case of $(\min, \max)$ structures. We first define the  tripartite negative triangle problem (also known as the $IJ$-bounded
triangle problem~\cite{WilliamsW10}).

\begin{problem}[{\sc tnt}: Tripartite negative triangle]\label{pb:tnt}
Given a tripartite graph $G=(I \cup J \cup K,E)$ with weights $w : E \rightarrow \IR$, a triangle $(i,j,k) \in I \times J \times K$ is called \emph{negative}
if $\max \{ w_{i,k} , w_{k,j}\} + w_{i,j} < 0$. The tripartite negative triangle
problem is to decide whether there exists a negative triangle in $G$.
\end{problem}

Note that the above definition is not symmetric in $I,J,K$: Only $I$ and $J$ are interchangeable.
Vassilevska and Williams proved the following reduction
(\cite[Theorem V.2]{WilliamsW10} for the special case where
$\odot = \max$).

\begin{theorem}\cite{WilliamsW10}\label{thm:reduction-WW}
Let $T(n)$ be a function such that $T(n)/n$ is nondecreasing. Suppose the tripartite negative triangle problem in an $n$-node graph can be solved in $T(n)$ time. Then the (min,max)-product of
two $n\times n$ matrices can be performed in
$O(n^2 T(n^{1/3}) \log W)$, where $W$ is the absolute value of the
largest integer in the output.
\end{theorem}

We now show how Problem~\ref{pb:tnt} reduces to Problem~\ref{pb:hyp}.
\begin{lemma}\label{lem:reduc}
The tripartite negative triangle problem {\sc (tnt)} reduces to
fixed-base hyperbolicity {\sc (hyp)} in quadratic time.
\end{lemma}

\begin{proof}
Besides {\sc tnt} and {\sc hyp}, we define two intermediate problems:

\begin{itemize}
\item Tripartite positive triangle ({\sc tpt}).
Given a tripartite graph $G=(I \cup J \cup K,E)$ with weights $w : E \rightarrow \IR$, decide if there exists a triangle 
$(i,j,k) \in I \times J \times K$ such that
$\min \{ w_{i,k} , w_{k,j}\} - w_{i,j} > 0$.
\item Positive triangle ({\sc pt}).
Given a complete undirected graph $G=(V,E)$ with weights $w : E \rightarrow \IR$ and $\alpha \geq 0$, decide if there are three distinct vertices $i,j,k \in V$ such that
$\min \{ w_{i,k} , w_{k,j}\} - w_{i,j} > \alpha$.
\end{itemize}

We will denote by $\leq_2$ the quadratic-time reductions. The lemma is obtained by the following sequence of reductions:
{\sc tnt} $\leq_2$ {\sc tpt} $\leq_2$ {\sc pt} $\leq_2$ {\sc hyp}.

Let us show that {\sc tnt} $\leq_2$ {\sc tpt}. Let $G=(I \cup J \cup K,E)$ and
$w : E \rightarrow \IR$ an instance of {\sc tnt}. Consider the same graph
with the following weights $w'$: $w'_{i,j} = w_{i,j}$ for $(i,j) \in I \times J$
and 
$w'_{\ell,k} = - w_{\ell,k}$ for $(\ell,k) \in (I \cup J) \times K$.
The graph $G$ with weights $w'$
has a tripartite positive triangle if there exists $(i,j,k) \in I \times J \times K$ such that
$\min \{ w'_{i,k} , w'_{k,j}\} - w'_{i,j} > 0$. This is equivalent
to $\max \{ w_{i,k} , w_{k,j}\} + w_{i,j} < 0$.

We now show that {\sc tpt} $\leq_2$ {\sc pt}.
Let $G=(I \cup J \cup K,E)$ and $w : E \rightarrow \IR$ an instance of
{\sc tpt}. Let $G'=(V',E')$ be the complete graph on the set of 
$3n$ nodes $V' = I \cup J \cup K$. We shall define a symmetric weight 
function $w':E' \rightarrow \IR$
and $\alpha\geq 0$  such that $G',w',\alpha$ is a positive instance of {\sc pt}
if and only if $G$ has a tripartite positive triangle. 
Let $\lambda = 1+\max_{e \in E} |w(e)|$.
For $(i,j) \in I \times J$, let $w'_{i,j} = w_{i,j}$. 
Let $w'_{\ell_1,\ell_2} = 2\lambda$
for $(\ell_1,\ell_2) \in (I \times I) \cup (J \times J) \cup (K \times K)$.
For $(\ell,k) \in (I \cup J) \times K$, let $w'_{\ell,k} = w_{\ell,k}+5\lambda$. 
Finally, let $\alpha = 5 \lambda$. Assume that $G$ contains a tripartite positive
triangle $(i,j,k)$, and thus $\min \{ w_{i,k}, w_{k,j}\} - w_{i,j} > 0$.
It means that $\min \{ w_{i,k} + 5 \lambda, w_{k,j} + 5 \lambda\} - w_{i,j} > 5 \lambda$.
So $\min \{ w'_{i,k}, w'_{k,j}\} - w'_{i,j} > \alpha$, and thus $G'$ has a
positive triangle. 

Conversely, assume $G'$ has a positive triangle $(i,j,k)$:
\begin{equation}\label{eq:inequality}
\min \{ w'_{i,k} , w'_{k,j}\} - w'_{i,j} > 5 \lambda.
\end{equation}
Let us determine the location of the triangle $(i,j,k)$.
\begin{itemize}
\item Inequality~\eqref{eq:inequality}
implies that $6 \lambda > \min \{ w'_{i,k} , w'_{k,j}\}
> 5 \lambda + w'_{i,j}$. Hence, $\lambda > w'_{i,j}$.
It follows that $(i,j) \in (I \times J) \cup (J \times I)$.
\item Inequality~\eqref{eq:inequality} also implies that
$\min \{ w'_{i,k} , w'_{k,j}\} > 5 \lambda + w'_{i,j} >  4 \lambda$.
Hence, it must be that both $(i,k)$ and $(j,k)$ belong to
$(I \times K) \cup (J \times K)$.
\end{itemize}
From the two conditions above, it follows that
$k \in K$ and $(i,j) \in (I \times J) \cup (J\times I)$.
Without loss of generality, let us assume
that $(i,j) \in I \times J$. Then  
$\min \{ w_{i,k} + 5\lambda , w_{j,k} + 5\lambda \} - w_{i,j} > 5 \lambda$
which means that $(i,j,k)$ is a tripartite positive triangle in $G$.

Let us show that {\sc pt} $\leq_2$ {\sc hyp}.
Consider a weighted complete graph $G$ on the vertices $\{1,\ldots,n\}$
and $\alpha \geq 0$ an instance of {\sc pt}. Let
$W$ be the (symmetric) weight matrix of $G$, so for any $i,j$ the coefficient 
$W_{ij}$ is the weight $w_{i,j}$ of the edge $(i,j)$, and $W_{ii}=0$. 
Let $\lambda'=|\max_{ij} w_{i,j}|$ be the absolute value of the largest weight.
Let $I_n$ be the $n \times n$
identity matrix and $E_n$ be the all-one matrix of size $n\times n$.
Let $P$ be the matrix defined by $P=W+ 2\lambda' E_n+ 4\lambda' I_n$. 
We will show that $(M,d)$ is a metric space, 
where $M=\{r,x_1,x_2,\ldots,x_n\}$, and the metric $d$ is defined as follows.
\[
\begin{array}{rlll}
d(x_i,x_j)&=& P_{ii} + P_{jj} -2 P_{ij} \quad & \text{for any $x_i,x_j \in M \setminus \{r\}$},\\
d(x_i,r)&= &d(r,x_i) \ = \ P_{ii} & \text{for any $x_i \in M \setminus \{r\}$},\\
d(x,x)&= &0 & \text{for any $x \in M$}.\\
\end{array}
\]

We now prove that $(M,d)$ is indeed a metric space. The function $d$ is 
symmetric because the matrix $P$ is symmetric. For any $x_i,x_j \in M \setminus \{r\}$, $i \neq j$,
we have $d(x_i,x_j)=8\lambda' - 2W_{ij}$ and thus $6\lambda' \leq d(x_i,x_j) 
\leq 10 \lambda'$. We also have $d(x_i,r)=P_{ii}=6\lambda'$, so  
$6 \lambda' \leq d(x,y) \leq 10 \lambda'$ for any two distinct points $x,y \in M$,
 which implies that the triangle inequality is satisfied. 

We have just proved that $(M,d)$ is a metric space. In addition, the matrix $P$ records
the Gromov products of $(M,d)$ at base $r$, because for any 
$x_i,x_j \in M \setminus \{r\}$, we have:
\[2(x_i\mid x_j)_r	 = d(x_i,r)+d(x_j,r)-d(x_i,x_j)
		 = P_{ii}+P_{jj}-(P_{ii}+P_{jj}-2P_{ij})
		 = 2P_{ij}.
\]

We now argue that for any $a,b,c \in M$ such that 
$\min\{(a \mid c)_r,(c \mid b)_r\}-(a\mid b)_r > 0$, we must have
$(a,b,c)=(x_i,x_j,x_k)$ for distinct $i,j,k \in \{1,\ldots,n\}$.
For any $y \in M$, we have $(y\mid r)_r=0$, and thus $r \notin \{a,b,c\}$. 
So we must have $(a,b,c)=(x_i,x_j,x_k)$ for some  $i,j,k \in \{1,\ldots,n\}$. Thus,
we only need to argue that $a,b$ and $c$ are distinct. 
Indeed, if  $a=c$, then
\[
\min\{(a \mid c)_r,(c \mid b)_r\}-(a \mid b)_r  = 
\min\{(a \mid a)_r,(a \mid b)_r \} -(a \mid b)_r
 \leq 0\]
and if $a=b$,
\begin{align*}
2(\min\{(a \mid c)_r,(b \mid c)_r\}-(a\mid b)_r)  & =
2(a \mid c)_r-2(a\mid a)_r \\  
 & = 2(a \mid c)_r-2d(a,r) \\
& = d(c,r)-d(c,a)-d(a,r) \leq 0.
\end{align*}

Thus, the hyperbolicity $\delta_r$ at base $r$ of $(M,d)$ is larger than $\alpha$ if and
only if there are distinct $i,j,k \in \{1,\ldots,n\}$ such that 
$\min\{(x_i \mid x_k)_r,(x_k \mid x_j)_r\}-(x_i\mid x_j)_r > \alpha$. So we showed that 
$\delta_r > \alpha$  if and only if there are distinct $i,j,k$ such that 
$\min\{P_{ik},P_{kj}\}-P_{ij} > \alpha$, or equivalently, 
$\min\{W_{ik},W_{kj}\}-W_{ij} > \alpha$, which means that $G$ has a positive triangle.
\end{proof}

\begin{theorem}
If fixed-base hyperbolicity can be decided in time $O(n^\nu)$, with $\nu \geq 2$,
then the (max,min)-product of two matrices of size $n \times n$ can be
done in time $O(n^{2+\nu/3} \log n)$.
\end{theorem}
\begin{proof}
Let us assume fixed-base hyperbolicity can be decided in time $O(n^\nu)$.
By Lemma~\ref{lem:reduc},
the tripartite negative triangle problem can also
be solved in time $O(n^\nu)$.
It remains to show that (max,min)-product can be done in time 
$O(n^{2+\nu/3} \log n)$. By duality, the complexity of computing
(max,min) products and (min,max) products are the same.
Moreover, we can assume without lost of generality that the two matrices
of which we want to compute the product have integer inputs in
the range $\{0,\ldots,2n^2\}$.
(Indeed, one can sort the inputs of the two matrices, replace input values with
with their ranks, perform the product, and
replace back the ranks with the initial values in the product.)
Applying Theorem~\ref{thm:reduction-WW} gives complexity
$O(n^{2+\nu/3} \log n)$ to compute this type of products.
\end{proof}

The above theorem immediately implies that any $O(n^\nu)$-time
algorithm with $\nu < 3(\omega -1)/2$ 
to solve fixed-base hyperbolicity would give an $O(n^\tau)$-time algorithm 
with $\tau<(3+\omega)/2$
for the (max,min)-product of two matrices of size $n \times n$.
Since the state of the art is $3(\omega -1)/2 > 2.05$, any $O(n^{2.05})$-time 
algorithm for fixed-base
Gromov hyperbolicity would yield an improvement on the complexity of
$(\max,\min)$-product.

%%%%%%%%%%%%%%%%%%%%%%%%%%%%%%%%%%%%%%%%%%%%%%%%%%%%%%%%%%%%%%%%%%%%%%%%%%%%%%%%%%%%%%%
%\bibliography{hyperbolicity}

\end{document}